\documentclass[times,3p,authoryear,colorlinks]{elsarticle}
\pdfoutput=1
\usepackage[utf8]{inputenc}
\usepackage[T2A,T1]{fontenc}
\usepackage{tempora} %
\usepackage[ngerman,russian,english]{babel}
\usepackage{microtype}
\usepackage{amsmath,amsthm,amssymb}
\usepackage{hyphenat}
\usepackage{doi}

\usepackage[capitalize,nameinlink]{cleveref}
\usepackage{todonotes}
\usepackage{paralist}
\usepackage{mathtools}
\newtheorem{theorem}{Theorem}

\newtheorem*{problemA}{Problem A}
\newtheorem*{problemB}{Problem B}
\newcommand{\rt}{tour}

\begin{document}
\begin{frontmatter}
  \title{
  A historical note
  on the 3/2-approximation algorithm\\
  for
  the metric traveling salesman problem}

  \author[mmf]{René van Bevern\corref{cor1}}
  \ead{rvb@nsu.ru}
  
  \author[ii,gi]{Viktoriia A.\ Slugina}
  \ead{v.slugina@g.nsu.ru}

  \cortext[cor1]{Correspondence to: Novosibirsk State University, ul.\ Pirogova 1, Novosibirsk, 630090, Russian Federation.}

  \address[mmf]{Department of Mechanics and Mathematics,
    Novosibirsk State University, Novosibirsk, Russian Federation}

  \address[ii]{Institute of History of the Siberian Branch of the Russian Academy of Sciences, Novosibirsk, Russian Federation}
  \address[gi]{Humanities Institute, Novosibirsk State University, Novosibirsk, Russian Federation}
  
  \begin{abstract}
    One of the most fundamental results in %
    combinatorial optimization %
    is the polynomial\hyp time 3/2-approximation algorithm
    for the metric
    traveling salesman problem.
    It was
    presented
    by Christofides in 1976
    and is
    well known as ``the Christofides algorithm''.
    Recently,
    some authors started
    calling it
    ``Christofides-Serdyukov algorithm'',
    pointing out that
    it was published independently
    in the USSR in 1978.
    We provide some historic background
    on Serdyukov's findings
    and a translation of his article from Russian into English.

    \begin{otherlanguage}{ngerman}
      \medskip
      \noindent
      \textbf{Zusammenfassung}

      \medskip
      \noindent
      Eines der grundlegendsten Resultate
      auf dem Gebiet der kombinatorischen Optimierung
      ist der Polynomialzeit\hyp 3/2\hyp Approximationsalgorithmus
      für das metrische Problem des Handlungsreisenden.
      Er wurde 1976 von Christofides vorgestellt
      und ist bekannt unter dem Namen ,,Christofides-Algorithmus''.
      In letzter Zeit bezeichnen ihn einige Autoren
      als ,,Christofides\hyp Serdyukov\hyp Algorithmus''
      mit dem Hinweis,
      er sei 1978 unabhängig in der UdSSR publiziert worden.
      In diesem Artikel beleuchten wir
      den historischen Hintergrund
      um Serdyukovs Entdeckung
      und liefern eine Übersetzung seines Artikels aus dem Russischen ins Englische.
    \end{otherlanguage}

    \begin{otherlanguage}{russian}
      \medskip
      \noindent
      \textbf{Аннотация}

      \medskip
      \noindent
      Одним из самых фундаментальных результатов
      в области комбинаторной оптимизации является
      полиномиальный 3/2-приближённый алгоритм для
      метрической задачи коммивояжёра.
      Он был представлен Никосом Кристофидесом в 1976\,г.\
      и хорошо известен под названием <<алгоритм Кристофидеса>>.
      В последнее время
      некоторые авторы стали
      называть его <<алгоритмом Кристофидеса-Сердюкова>>,
      утверждая,
      что он был опубликован независимо в СССР в 1978\,г.
      В этой статье рассматривается
      исторический контекст появления результата А.\ И.\ Сердюкова
      и приводится перевод его статьи с русского на английский язык.
    \end{otherlanguage}
  \end{abstract}
\end{frontmatter}
\section{Introduction}
\noindent
One of the most fundamental problems
in combinatorial optimization
is the traveling salesman problem,
formalized as early as 1832
\citep[cf.][Chapter~1]{ABCC06}:
given $n$~cities
and their pairwise distances,
find a shortest \rt{}
to visit each city exactly once
and return to the starting point.

Finding the \emph{shortest} \rt{}
is a computationally intractable problem
even in the special case
where the distances between the cities
satisfy the triangle inequality \citep{GJ79}.
\citet{Chr76} 
presented an $O(n^3)$-time \emph{3/2\hyp approximation algorithm}
for this special case:
it yields a \rt{}
that is at most 3/2 times longer
than the shortest one.
It %
is
a prime example for approximation algorithms
that entered
textbooks
and encyclopedias as
``the Christofides algorithm''
or ``the Christofides heuristic''
\citep{GJ79,Chr79,Gut09,WS11,Bla16}.

Quite some efforts have been made
trying to improve it
\citep[cf.\ the surveys of][]{Vyg12,Sve13}.
One line of research aims for improving its running time:
there are many faster heuristics,
which
cannot guarantee 3/2-approximate solutions \citep{JM07},
yet $(3/2+\varepsilon)$-approximate solutions for
any~$\varepsilon>0$
are computable
by a randomized algorithm
in $O(n^2\log^4n/\varepsilon^2)$~time \citep{CQ17}.
Another line of research aims for improving the approximation factor,
which was successful only in special cases:
in polynomial time,
one can compute 8/7-approximate solutions if the distances are in~$\{1,2\}$
\citep{BK06},
7/5-approximate solutions if
the distances are lengths of shortest paths
in an unweighted graph \citep{SV14},
and $(1+\varepsilon)$\hyp approximate solutions
for any fixed~$\varepsilon>0$
if the cities are points in fixed\hyp dimensional
Euclidean \citep{Aro98,Mit99} or
doubling spaces \citep[using a randomized algorithm]{BGK16},
or if the distances are lengths of shortest paths
in a graph excluding some fixed minor \citep{DHK11}.

For general distances satisfying the triangle inequality,
the 3/2 approximation factor of \citeauthor{Chr76}' algorithm
remains the state of the art.\footnote{Actually,
  \citet{Wol80} showed
  that the length of the computed tours
  is within a factor 3/2
  not only of the optimum,
  but
  of a lower bound
  given by the optimal solution to a relaxation
  of an integer linear programming model.}
Recently,
a small but growing group of authors started referring
to it
as ``the Christofides-Serdyukov algorithm''
\citep{BDW98,DT10,BBMV19,GW19,SZ19,Tar19,TV19},%
\footnote{We deliberately omit articles coauthored by
  Serdyukov's former colleagues from this list.}
claiming that it was independently obtained
in the USSR by \citet{Ser78}.

\looseness=-1
At the one hand,
this claim is plausible:
in the beginning of the 70s,
a~lot of research on computationally intractable
problems was carried out parallely in the  USSR,
leading to independent proofs of seminal results
like the Cook-Levin theorem
about the NP\hyp completeness
of the satisfiability problem for Boolean formulas
\citep{Tra84}.
Moreover,
the submission date
given in
the journal article of \citet{Ser78},
January 27th, 1976,
predates
the report of \citet{Chr76}, dated February, 1976.
On the other hand,
such claims should be treated with caution:
for example,
the wide\hyp spread claim
that Kuratovski's theorem
was earlier proved in the USSR
has little support \citep{KQS85}.

We give some historic background
on \citeauthor{Ser78}'s findings,
which indeed supports the claim
of his independent discovery of the 3/2\hyp approximation algorithm
and sheds some light on the timely coincidence
of the publications of \citeauthor{Chr76} and \citeauthor{Ser78}.
We also provide a translation of \citeauthor{Ser78}'s article
in the appendix.

\section{Anatoliy I.\ Serdyukov (1951--2001)}
\noindent
The following information
about Serdyukov
can be found in \citet{BGS07}, \citet{Dem17},%
\footnote{The birth year 1952 given in the book edited by \citet{Dem17}
  is incorrect.}
the archives
of the Department of Mechanics and Mathematics
of Novosibirsk State University,
and the
State Public Scientific Technological Library
of the Siberian Branch of the Russian Academy of Sciences.

Anatoliy Ivanovich Serdyukov was born
on October 29th, 1951,
in Prokopyevsk,
a city in Kemerovo region (Western Siberia), USSR.
He graduated from
Novosibirsk State University in 1973,
after which
he was employed
in the structures of the Siberian Branch
of the Lenin Academy of Agricultural Sciences,
then at the
Institute of Cytology and Genetics
of the Siberian Branch of the Academy of Sciences of the USSR
(SB AS USSR),
and finally
at the Institute of Mathematics of the SB AS USSR
(now named the Sobolev Institute of Mathematics,
Siberian Branch of the Russian Academy of Sciences),
where he was working
until his death on February 7th, 2001.

In 1980,
already working at the Institute of Mathematics,
he was awarded the academic degree
of candidate of physico\hyp mathematical sciences.
\citeauthor{Ser80}'s \citeyearpar{Ser80} thesis
is on the complexity of finding Hamiltonian and Eulerian cycles in graphs.
His best known results are
approximation algorithms
for finding \emph{longest}
traveling salesman tours \citep[surveyed by][]{BGS07}.

Taking into account his graduation year
and the submission date of \citeauthor{Ser78}'s \citeyearpar{Ser78}
article, January 27th, 1976,
\citeauthor{Ser78}
must have obtained his 3/2\hyp approximation algorithm
as a young graduate student in about 1975.

\section{Circulation of Christofides' result between 1976 and 1979}
\noindent
Authors usually refer to \citeauthor{Chr76}' \citeyear{Chr76}
technical report at Carnegie-Mellon University (CMU)
as the source of the 3/2\hyp approximation algorithm
for the metric traveling salesman problem,
which some authors
do not consider as published
\citep[cf.][who also claims that ``Christofides never published his algorithm'']{Bla16}.

\looseness=-1
Apparently,
\citeauthor{Chr76}' technical report was not known
to a wide audience up to 1978. %
For example,
\citet{Kar77} and \citet{RSL77}
refer to Christofides' abstract
in the proceedings
of a symposium held at CMU in April 1976.
The proceedings
were published only in December 1976 \citep{Tra76}.
\citet{FHK76}
refer to the same abstract,
whereas later,
in the journal version of their article,
\citet{FHK78} refer to the technical report.
\citeauthor{Chr76}' technical report could have been popularized
in 1977,
when its abstract
stating the 3/2-approximation
was indexed by the NASA abstract journal
Scientific and Technical Aerospace Reports
\citep{Chr77}.

Some authors of that time,
for example \citet{LR79},
refer to a journal article of Christofides
that is to appear in the journal \emph{Mathematical Programming}.
In a combinatorial optimization textbook,
\citet{Chr79}
describes his algorithm
without proving the approximation factor,
referring to an article in press in \emph{Mathematical Programming}
for the proof,
not mentioning his technical report.
Interestingly,
according to the archives of \emph{Mathematical Programming},
his article was not published.
The algorithm with complete proof details
was published to a wide audience not later
than in the seminal textbook of \citet{GJ79}.

Summarizing,
\citet{Ser78} submitted his journal article in January 1976,
which predates all traces of \citeauthor{Chr76}' publications on this topic.
Thus,
it is plausible that \citet{Ser78}
obtained the result independently.

\section{Serdyukov's work between 1974 and 1978}
\noindent
\looseness=-1
We give some historic background
on the findings of \citeauthor{Ser78}
to shed some light on the
timely coincidence of
the publications
of \citet{Chr76} and \citet{Ser78}.
To this end,
it is helpful to interpret the 3/2-approximation algorithm
for the traveling salesman problem as follows:
A first step
computes a minimum\hyp cost spanning tree
that connects all the cities.
A second step
computes a shortest \rt{} in the input graph
that traverses the edges of the spanning tree.

The second step is solved
using an approach earlier developed for the closely
related Chinese postman problem
of computing a shortest \rt{}
traversing \emph{all} edges of a graph:
\citet{Chr73} and \citet{Ser74},
but also \citet{EJ73}, actively
study the Chinese postman problem at that time.
They all reduce it
to
the problem of finding a minimum\hyp cost perfect matching
on the
complete edge\hyp weighted graph
on all odd\hyp degree vertices of the input graph.%
\footnote{Notably,
\cite{Ser74} explicitly introduces
the  problem
that forty years later
is intensively studied as the Eulerian extension problem
\citep{SBNW11,HJM12,SBNW12,DMNW13,GWY17,BFTxx}.}
Surprisingly,
while 
\citeauthor{Chr73}, \citeauthor{EJ73}
solve the matching problem
using the polynomial\hyp time algorithm of \citet{Edm65b},
\citeauthor{Ser74}
reduces it
to an exponential number of matching problems in bipartite graphs.%
\footnote{In contemporary terms of parameterized complexity theory
  \citep[cf.][]{CFK+15},
  \cite{Ser74} merely describes a fixed\hyp parameter algorithm
for the Chinese postman problem
parameterized by the number of odd\hyp degree vertices in the input graph.}
Apparently,
in 1974
neither
\citeauthor{Ser74}
nor his reviewers
were aware
of the work of \citet{Chr73}, \citet{EJ73},
or
the
polynomial\hyp time
algorithm for computing maximum\hyp weight
matchings in general graphs,
published by \citeauthor{Edm65b} nine years earlier.

Since \citet{Ser78} uses \citeauthor{Edm65b}' algorithm
to solve the matching problem
in his 3/2\hyp approximation algorithm
for the traveling salesman problem
but was unaware of it in 1974,
he must have learned about \citeauthor{Edm65b}' algorithm
in 1974 or 1975.
One scenario is
that he learned about it
via the article of \citet{Chr73},
which \citet{Ser76} cites
in an article studying
reductions between matching, covering,
the Chinese postman, and the traveling salesman problems.
In this scenario,
\citeauthor{Ser78} obtained
his 3/2\hyp approximation
independently of \citeauthor{Chr73}
but because of him.
Another scenario is that
\citet{Ser78} learned about \citeauthor{Edm65b}' algorithm
from \citet{Kar76},
whose $O(n^3\log n)$\hyp time implementation
of \citeauthor{Edm65b}' algorithm
he uses in his 3/2\hyp approximation.
\citeauthor{Kar76}'s
article was probably not yet published in January 1976,
when \citeauthor{Ser78} submitted his article,
but he might have had access to a preliminary copy,
which is supported by the fact
that the titles given by \citeauthor{Ser78} and \citeauthor{Kar76}
differ slightly.

\section{Conclusion}
\noindent
Our findings support the claim that
\citet{Ser78} discovered the 3/2\hyp approximation algorithm
for the metric traveling salesman problem
independently of \citet{Chr76}.

Concerning the timely coincidence
of the publications of \citeauthor{Chr76} and \citeauthor{Ser78},
we conclude that,
on the one hand,
it was impossible for
\citeauthor{Ser78}
to find the algorithm
much earlier
than \citeauthor{Chr76},
being unaware of \citeauthor{Edm65b}'
polynomial-time matching algorithm
up to 1974.
On the other hand,
actively working on
the Chinese postman before,
he found the 3/2\hyp approximation
for the traveling salesman problem
as soon as he became aware of
\citeauthor{Edm65b}' algorithm.

\looseness=-1
An English abstract of \citeauthor{Ser78}'s \citeyearpar{Ser78}
article
was indexed in zbMATH only in 1982 \citep{Ser82}.
At~this time,
``the \citeauthor{Chr76} algorithm''
had already entered
fundamental textbooks like that of \citet{GJ79}.
Moreover,
the English abstract
does not mention any approximation factors.
Thus,
it is not surprising that \citeauthor{Ser78}'s result
remained largely unknown beyond the USSR.

\paragraph{Acknowledgments}
We thank Edward Kh.\ Gimadi and Oxana Yu.\ Tsidulko
for helpful input.

\paragraph{Funding}
René van Bevern is supported by the
Mathematical Center in Akademgorodok,
agreement No.\ 075-15-2019-1675 with the Ministry of Science and
Higher Education of the Russian Federation.
Viktoriia A.\ Slugina is supported by
grant No.\ 19-39-60006
of the Russian Foundation for Basic Research.

\bibliographystyle{tsp-history}
\bibliography{tsp-history}

\newpage
\appendix
\section*{Appendix}
\noindent
The following
is a translation of
the Russian article of 
A.\ I.\ Serdyukov,
O nekotorykh ekstremal’nykh obkhodakh v grafakh,
\emph{Upravlyaemye sistemy} 17:76--79, 1978.

\section*{On some extremal walks in graphs}
\noindent
A.\ I.\ Serdyukov

\bigskip
\noindent
\textbf{1.}\quad
Let $G=(X,U)$ be an undirected $n$-vertex graph
with edge weights $\rho_{ij},u_{ij}\in U,1\leq i,j\leq n$,
where the $\rho_{ij}$ are positive real numbers.
Moreover,
let there be a real number~$k\geq 1$.
By $\mathcal L$ we denote the set of Hamiltonian cycles
in the graph~$G$,
and by $\mathcal M$
the set of cycles
that contain all vertices of~$G$.
We consider two extremum problems.

\begin{problemA}
Find an element $L^*\in\mathcal L$ with the property
\[
  \rho(L^*)=\smashoperator{\sum_{u_{ij}\in L^*}}\rho_{ij}\leq k\cdot\min_{L\in\mathcal L}\rho(L).
\]
\end{problemA}
\noindent
Note that,
for $k=1$,
Problem A is nothing else but the traveling salesman problem.

\begin{problemB}
Find an element $M^*\in\mathcal M$ with the property
\[
    \rho(M^*)=\smashoperator{\sum_{u_{ij}\in M^*}}\rho_{ij}\leq k\cdot\min_{M\in\mathcal M}\rho(M).
\]
(herein we take into account the multiplicity
that each edge appears on the cycle).
\end{problemB}
\noindent
For $k=1$,
Problems A and B are NP\hyp complete~[1].
Moreover, as shown in [2],
Problem A is NP\hyp complete for arbitrary~$k\geq 1$.
In [3], an algorithm with polynomial running time
is presented for Problem A with $k=2$
in complete undirected graphs
whose edge weights respect the triangle inequality.

In this work, we present an algorithm for Problem A with $k=3/2$
in the very same class of graphs,
whose complexity is $O(n^3\ln n)$ operations.

Terminology related to graph theory used in this work
can be found in [4].

\bigskip
\noindent
\textbf{2.}\quad
Assume that the edge weights
of a complete undirected $n$-vertex graph~$G$
satisfy the triangle inequality:
\begin{align*}
  \label{eq:tri}
  \rho_{ij}\leq\rho_{ik}+\rho_{kj},\quad 1\leq i,j,k\leq n.
  \tag{1}
\end{align*}
We introduce some necessary notation:
\begin{compactitem}[$\bar G=(X_1,\bar U)$ ---]
\item[$L_0$ ---] a minimum Hamiltonian cycle in graph~$G$;
\item[$M_0$ ---] the shortest cycle containing all vertices of graph~$G$;
  
\item[$\mathcal D_0$ ---] the shortest spanning tree in graph~$G$;
  
\item[$X_1\subseteq X$ ---] the set of odd\hyp degree vertices
  in the tree~$\mathcal D_0$;
  
\item[$\bar G=(X_1,\bar U)$ ---] the complete subgraph of~$G$
  with vertex set~$X_1$;
\item[$\bar L_0$ ---] a minimum Hamiltonian cycle in~$\bar G$;
  
\item[$w_0$ ---] a minimum perfect matching in graph~$\bar G$.
\end{compactitem}
We prove the following theorem.

\begin{theorem}
  Independently of the edge weights of the input graph~$G$,
  which satisfy (1),
  the following relations hold:
  \begin{align*}
    \rho(L_0)&=\rho(M_0),\tag{2}\\
    \rho(L_0)&\geq\rho(\bar L_0),\tag{3}\\
    \rho(w_0)&\leq \frac12 \rho(\bar L_0).\tag{4}
  \end{align*}
\end{theorem}
The proof of equality (2) can be found in [5]
(cf.\ Lemma 1).
Inequality (4) follows from the fact that
graph~$\bar G$ has an even number of vertices
(since the number of odd\hyp degree vertices in a graph is even [4])
and that a Hamiltonian cycle can be partitioned
into two edge\hyp disjoint perfect matchings.
We prove inequality (3).
To this end,
we write the Hamiltonian cycle~$L_0$
as a sequence of vertices:
\begin{align*}
  \{x_1,x_2,x_3,\dots,x_{n-1},x_n,x_1\}.\tag{5}
\end{align*}
If $X_1=X$,
then inequality (3) is proven.
Let $X_1\subsetneq X$.
In this case, there must be vertices
$x_i\in X_1,x_j\in X_1,i<j-1$
such that $x_{i+s}\notin X_1,1\leq s<j-i$.
Then,
replacing the path between vertices~$i$ and~$j$ in the Hamiltonian cycle~$L_0$
by the edge~$u_{ij}$,
we obtain a simple cycle
visiting all vertices in~$X_1$
with a weight not exceeding~$\rho(L_0)$.
This is possible since the graph~$G$ is complete
and (1) holds.
By executing this process for all pairs of such vertices
in the sequence (5),
we obtain a Hamiltonian cycle in the graph~$\bar G$
of weight not exceeding $\rho(L_0)$.
Theorem 1 is proved.

\bigskip
\noindent
\textbf{3.}
\quad
We now describe the algorithm for Problem A
in the class of complete undirected graphs
whose edge weights satisfy (1).
The algorithm consists of five steps.

\textit{First step.}
Compute a minimum spanning tree~$\mathcal D_0$ in the graph~$G$.
To this end,
one can use the algorithm of Prim~[6],
whose complexity is $O(n^2)$~operations.

\textit{Second step.}
Find all odd\hyp degree vertices in the tree~$\mathcal D_0$,
which form the set $X_1\subseteq X$.
Then build the subgraph~$\bar G=(X_1,\bar U)$ described above.

\textit{Third step.}
Find the perfect matching~$w_0$ in graph~$\bar G$.
To find such a perfect matching,
one can use the algorithm described
in [7],
whose complexity is $O(n^3\ln n)$~operations.
(In fact, the algorithm in [7] solves the problem
of finding a maximum\hyp weight matching in a graph.
To reduce the problem of finding a minimum\hyp weight
perfect matching in the graph~$\bar G$
to this problem,
it is enough to assign the edges of~$\bar G$
the weights $\rho^*_{ij}=2a-\rho_{ij},\bar u_{ij}\in\bar U$,
where $a=\max_{\bar u_{ij}\in\bar U}\rho_{ij}.)$
Note that the edge set $\mathcal D_0\cup w_0$
forms an Eulerian graph
(taking into account multi\hyp edges
that may appear when $\mathcal D_0\cap w_0\ne\emptyset$).

\textit{Fourth step.}
Find an Eulerian walk of all edges of the graph~$\mathcal D_0\cup w_0$
using the algorithm described in [4].

\textit{Fifth step.}
Write the Eulerian walk for the edges in the set $\mathcal D_0\cup w_0$
as a sequence of vertices
\begin{align}
  \{x_1,x_2,x_3,\dots,x_{k-1},x_1,x_k,\dots,x_1\}.\tag{6}
\end{align}
\looseness=-1
If there are no vertex repetitions in this sequence,
then $\mathcal D_0\cup w_0$
is a Hamiltonian cycle.
Then we set $L^*=D_0\cup w_0$.
In this case the algorithm stops.
Assume that some vertex,
without loss of generality~$x_1$,
is repeated in the sequence~(6).
Then we transform (6) into a cycle $\{x_1,x_2,x_3,\dots,x_{k-1},x_k,\dots,x_1\}$
without increasing its weight.
This is possible because $G$~is a complete graph
whose edge weights satisfy (1).
Executing this process
for each vertex of graph~$G$
as often as it is repeated in (6),
we obtain a Hamiltonian cycle~$L_1$.
Set $L^*=L_1$.
This concludes the description of the algorithm.

\begin{theorem}
  The weight of the Hamiltonian cycle~$L^*$
  constructed by the described algorithm
  in the graph~$G$ satisfies the inequality
  \begin{align*}
    \rho(L^*)<\frac32\cdot \rho(L_0).\tag{7}
  \end{align*}
\end{theorem}
\begin{proof}
  To prove this,
  it is enough to notice that
  $\rho(L^*)\leq\rho(\mathcal D_0)+\rho(w_0)<\rho(L_0)+\rho(w_0)$
  and to use Theorem~1.
\end{proof}

\noindent
\textbf{4.}
\quad
We now consider Problem B in an arbitrary connected graph~$\tilde G=(X,\tilde U)$,
the edges of which have weights~$\tilde\rho_{ij},\tilde u_{ij}\in\tilde U,1\leq i,j\leq n$,
where $\tilde\rho_{ij}$ are positive real numbers.
For $\tilde G$
we build the shortest\hyp path graph~$G=(X,U)$,
which is complete and whose edge weights satisfy (1)
(cf. [5]).
By $\tilde M_0$ denote the shortest cycle
containing all vertices of~$\tilde G$.
Then we have the following equality [5]:
\begin{align*}
  \tilde\rho(\tilde M_0)=\rho(L_0).\tag{8}
\end{align*}
By replacing each edge $u_{ij}\in L^*$
by a shortest path between the vertices~$i$ and~$j$ in the graph~$\tilde G$,
we find some cycle~$\tilde M^*$
containing all vertices of graph~$\tilde G$,
whose weight equals $\rho(L^*)$.
Moreover,
taking into account (7), (8),
we obtain the following inequality:
\[
  \tilde \rho(\tilde M^*)< \frac32\cdot \tilde\rho(\tilde M_0).
\]
This last inequality means that one can use $\tilde M^*$
as a solution to Problem B in the graph~$\tilde G$.

For the complexity analysis of the algorithm,
it is enough to note
that each of its five steps requires no more than $Cn^3\ln n$ operations.
All subsequent computations related to finding a solution for Problem~B
can be executed using the algorithm of Dijkstra [8] in $O(n^3)$~operations.
Thus,
the number of operations required for realizing the described algorithms
is proportional to $n^3\ln n$.

\hfill Received by the editorial-publishing office

\hfill January 27, 1976
\subsection*{Literature}

\begin{compactenum}[1.]
\item  Karp R.\ M.,  Slozhnost' resheniya ekstremal'nykh kombinatornykh problem. Kiberneticheskii sbornik, volume 12, Moscow, Mir, 1975.
  
\item Sahni S.\ and Gonzales T., P-complete approximation problems. Journal of Association for Computing Machinery, July 1976, No.\ 3, volume 23, pp.\ 555--565.

\item Rosenkrantz D.\ J., Stearns R.\ E.\ and Lewis P.\ M.,
  Approximate algorithms for the travelling salesperson problem.
  15th Annual IEEE Symp.\ on Switching and Automata Theory,
  1974, pp.\ 33--42.

\item Berzh K., Teoriya grafov i ee prilozhenie. Moscow, IL, 1962.

\item Serdyukov A. I., O vzaimnoi svodimosti nekotorykh ekstremal'nykh zadach teorii grafov.  Upravlyaemye sistemy, No.\ 15, Novosibirsk, 1976, pp.\ 68--73.

\item Prim R.\ M.,
  Kratchaishie svyazyvayushchie seti i nekotorye obobshcheniya.
  Kiberneticheskii sbornik, volume 2, Moscow, Mir, 1961.

\item Karzanov A.\ V., Ekonomnye realizatsii algoritmov Edmondsa nakhozhdeniya parosochetanii maksimal'noi moshchnosti (maksimal'nogo vesa).
  In Issledovaniya po Diskretnoi Optimizatsii, Moscow, Mir, 1976,
  pp.\ 306--327.
  
\item Khu T., Tselochislennoe programmirovanie i potoki v setyakh.  Moscow, Mir, 1974.
\end{compactenum}
\end{document}